\documentclass[a4paper]{article}
\usepackage{amsfonts,amssymb,amsmath,amsthm,cite}
\usepackage{ucs} 
\usepackage[utf8x]{inputenc}
\usepackage{graphicx}
\usepackage[english]{babel}
\usepackage{slashed}
\usepackage{textcomp}

\evensidemargin=\oddsidemargin
\newcommand{\MM}{\mathcal{M}}

\newcommand{\DD}{\mathcal{D}}

\newtheorem{theorem}{Theorem}

\newtheorem{lem}{Lemma}

\newtheorem*{definition}{Definition}

\begin{document}
\begin{center}
	
	{\Large\textbf{Index Theorem for Domain Walls}}
	\vspace{0.5cm}
	
	{\large A.~V.~Ivanov$^\dag$}
	
	\vspace{0.5cm}
	
	$^\dag${\it St. Petersburg Department of Steklov Mathematical Institute of
		Russian Academy of Sciences,}\\{\it 27 Fontanka, St. Petersburg 191023, Russia}\\
	$^\dag${\it Leonhard Euler International Mathematical Institute, 10 Pesochnaya nab.,}\\{\it St. Petersburg 197022, Russia}\\
	{\it E-mail: regul1@mail.ru}
	
\end{center}

\begin{abstract}
The paper is devoted to the discussion of index theorem for domain walls condition. We give an extension of the theorem to the case, when not only Yang--Mills connection components have a jump on some surface of co-dimension 1, but also components of a Riemannian connection, while a metric remains continuous. 
\end{abstract}

\section{Introduction}
The Atiyah--Patodi--Singer index theorem (see \cite{Atiyah:1975jf}) plays a crucial role in the modern mathematical physics. It relates a spectral function of a Dirac operator and the integral of the corresponding Pontryagin density. The last one can contain singular terms that lead to consideration of surface terms on the manifold. Using the language of theoretical physics this can be expressed as follows: the index theorem gives relations between bulk and boundary anomalies (see \cite{Witten:2019bou}). 

One of the tasks of recent works is to generalize the theorem to the case of domain walls (see \cite{Fukaya:2019qlf,j+8}). For instance, it can be defined as the case when the gauge potential has a jump on some co-dimension 1 surface of the manifold. Actually, this problem is well-posed and you can get an explicit formula (see \cite{ivvas} and formula (\ref{APST}) below), which has a rather cumbersome form.

The first part of this work is devoted to the discussion of the second and third terms on the right hand side of formula (\ref{APST}). Indeed, they are integrals over the surface and at first glance have a different nature. The second one is the relative spectral asymmetry (RSA) for Dirac operators on both sides of the surface, while the third term does not make obvious sense and contains an extrinsic curvature.

In the paper we give a definition of the generalized RSA, which allows you to combine the surface terms into one in a quite natural way. We also prove the consistency of the new definition with the previous one and derive new relations (see Theorem \ref{th1}). 

Then we consider an extention of the Atiyah--Patodi--Singer index theorem to the case when both the Yang--Mills and the Riemannian connection components have discontinuities, while the metric tensor is continuous (see Theorem \ref{th2}). An example of using such type of condition can be found in \cite{vas}.

We are not going to repeat the long introduction of the work \cite{ivvas}, but only give a schematic statement of the problem, omitting the subtleties. Actually, information from the next section is enough to understand the theorems, but the reader can also refer to the section 2 from the paper \cite{ivvas}. Most of notations are borrowed from the monograph \cite{Nakahara:2003nw}.

\section{The main result}
Let us briefly recall the result of the paper \cite{ivvas}. First of all we introduce a compact Riemannian even-dimensional manifold $\MM$ without a boundary and its closed submanifold $\Sigma$, such that $\dim\MM=n$ and $\dim\Sigma=\mathrm{n}-1$. Let $g$, $\Gamma$, and $\mathcal{R}$ denote a smooth metric, the Riemannian connection 1-form, and the curvature 2-form. 

Then, after the introduction of an Hermitian vector bundle over $\MM$ and the choice of a special type of gauge fixing and representation for the Clifford structure, we can define a Yang--Mills connection 1-form $\mathcal{A}^+$, the corresponding curvature 2-form $\mathcal{F}^+$, and  a Dirac-type operator (see \cite{BV1}) on $\MM$, which near $\Sigma$ has the following structure
\begin{equation*}
\slashed{D}=\left( \begin{array}{cc} 0 & -\partial_{s} + \DD +\tfrac 12 K_{a}^{a} \\
\partial_{s} +\DD-\tfrac 12 K_{a}^{a} & 0 \end{array} \right), 
\label{DDgen}
\end{equation*} 
where $\DD$ is a Dirac operator on $\Sigma$, $K_{a}^{a}$ is the trace of an extrinsic curvature on $\Sigma$, and $s$ is the geodesic distance to $\Sigma$.

The domain walls condition means, that in some neighborhood of $\Sigma$ the form $\mathcal{A}^+$ is represented as a sum $\mathcal{A}^-+\theta(s)\mathcal{A}$, where $\theta(s)$ is the Heaviside step function, and $\mathcal{A}^-$ and $\mathcal{A}$ have smooth densities. Let $\mathcal{F}^-$ be the curvature 2-form for $\mathcal{A}^-$. Actually, the operator $\DD$ depends on $\mathcal{A}^+$, so if the form has a jump on $\Sigma$ the operator $\DD$ has different values from the both sides, $\DD^+$ and $\DD^-$.

The main result of the paper \cite{ivvas} is
\begin{equation}
\label{APST}
\mathrm{Index}(\slashed{D})=\int_{\MM\backslash\Sigma}P - \tfrac 12 \tilde\eta(\DD^+,\DD^-)-
\int_\Sigma T\hat{A}
(\Gamma^{\prime},\Gamma )\wedge [\mathrm{ch}(\mathcal{F}^+)-
\mathrm{ch}(\mathcal{F}^-)],
\end{equation}
where $P$ is the Pontryagin $n$-form, $\tilde\eta(\DD^+,\DD^-)$ is the RSA of the operators $\DD^+$ and $\DD^-$, $\mathrm{ch}$ is the Chern polynomial, and $T\hat{A}$ is the transgression of the $\hat{A}$-genus. Then, $\Gamma^{\prime}$ is a modification of the Riemannian connection 1-form $\Gamma$, that corresponds to a metric $g^{\prime}$, regularized version of $g$ in some neighbourhood of $\Sigma$, such that its  components have the property $g^{\prime}_{\mu\nu}=g^{\phantom{1}}_{\mu\nu}|_{s=0}$ near the surface $\Sigma$ (for details see section 4 of the paper \cite{ivvas}).

The standard definition of the RSA is formulated by using a Seeley--DeWitt coefficient for a family of Dirac operators $\DD(s)$, such that $\DD(0)=\DD^-$ and $\DD(1)=\DD^+$:
\begin{equation}
\tilde\eta(\DD^+,\DD^-)= -\frac 2{\sqrt{\pi}} \int_0^1 ds\, a_{n-2}(\partial_s\DD(s),\DD(s)^2).\label{tildeta}
\end{equation}

But this definition omits some information. Actually the surface $\Sigma$ is the submanifold of $\MM$ and has the extrinsic curvature. Roughly speaking, the definition (\ref{tildeta}) does not take into account the \textquotedblleft background\textquotedblright. 
To move on, let us expand the problem statement. We assume that on the surface $\Sigma$ not only the gauge connection $\mathcal{A}^+$ has a jump, but also the Riemannian one, while the metric tensor remains continuous. So locally near $\Sigma$ we have a decomposition $\Gamma^+=\Gamma^-+\theta(s)\Gamma_0$. By $\mathcal{R^+}$ and $\mathcal{R^-}$ we denote the corresponding Riemannian curvatures. Now we are ready to present our definition of a generalized relative spectral asymmetry
\begin{definition}
\label{def1}
Using the previous constructions and assumptions the generalized RSA equals to
\begin{equation}
\label{rsa1}
\tilde{\eta}(\mathcal{A}^+,\Gamma^+,\mathcal{A}^-,\Gamma^-)=
-2\int_{\Sigma}\bigg(\hat{A}(\mathcal{R}^+)\wedge T
\mathrm{ch}(\mathcal{A}^+,\mathcal{A}^-)+
T\hat{A}(\Gamma^+,\Gamma^-)\wedge\mathrm{ch}(\mathcal{F}^-)\bigg),
\end{equation}
where the restriction of the forms to the surface $\Sigma$ is meant.
\end{definition}
It is easy to obtain an equivalent representation for the generalized RSA
\begin{equation}
\label{rsa2}
\tilde{\eta}(\mathcal{A}^+,\Gamma^+,\mathcal{A}^-,\Gamma^-)
=-2\int_{\Sigma}\bigg(\hat{A}(\mathcal{R}^-)\wedge T
\mathrm{ch}(\mathcal{A}^+,\mathcal{A}^-)+
T\hat{A}(\Gamma^+,\Gamma^-)\wedge\mathrm{ch}(\mathcal{F}^+)\bigg),
\end{equation}
by using Stokes' theorem and the following equalities for the transgression:
\begin{equation*}
dT\mathrm{ch}(\mathcal{A}^+,\mathcal{A}^-)=
\mathrm{ch}(\mathcal{F}^{+})-\mathrm{ch}(\mathcal{F}^{-}),\,\,\,\,
dT\hat{A}(\Gamma^+,\Gamma^-)=\hat{A}(\mathcal{R}^{+})-\hat{A}(\mathcal{R}^{-}).
\end{equation*}

If the Riemannian connection 1-form does not have the jump ($\Gamma^-=\Gamma^+$), then we get a particular case
\begin{equation}
\label{rsa3}
\tilde{\eta}(\mathcal{A}^+,\Gamma^+,\mathcal{A}^-,\Gamma^+)=
-2\int_{\Sigma}\hat{A}(\mathcal{R}^+)\wedge T
\mathrm{ch}(\mathcal{A}^+,\mathcal{A}^-).
\end{equation}

To understand why the definition looks like this, let us try to study it in the context of the Atiyah--Patodi--Singer index theorem for domain walls. This is convenient to do in two stages. First, we consider the case studied earlier, when the density of the Riemannian connection 1-form is a smooth function. Hence, $\Gamma^+=\Gamma^-=\Gamma$. Also let $\Gamma^{\prime}$ be the 1-form, constructed for $\Gamma$, from formula (\ref{APST}).
\begin{theorem}
\label{th1}
Under the conditions described above we have the equalities:
\begin{equation}
\label{RSA}
\tilde\eta(\mathcal{A}^+,\Gamma^{\prime},\mathcal{A}^-,\Gamma^{\prime})=\tilde\eta(\DD^+,\DD^-);
\end{equation}
\begin{equation}
\label{for1}
\tilde\eta(\mathcal{A}^+,\Gamma,\mathcal{A}^-,\Gamma)=\tilde\eta(\DD^+,\DD^-)+2
\int_\Sigma T\hat{A}
 (\Gamma^{\prime},\Gamma )\wedge [\mathrm{ch}(\mathcal{F}^+)-
 \mathrm{ch}(\mathcal{F}^-)];
\end{equation}
\begin{equation}
\label{APST1}
\mathrm{Index}(\slashed{D})=\int_{\MM\backslash\Sigma}P - \tfrac 12 \tilde\eta(\mathcal{A}^+,\Gamma,\mathcal{A}^-,\Gamma).
\end{equation}
\end{theorem}
\begin{proof} 
It is easy to note, that formula (\ref{APST1}) follows from formula (\ref{for1}) and the Atiyah--Patodi--Singer index theorem for domain walls (\ref{APST}).

Then formula (\ref{for1}) can be obtained by using formula (\ref{RSA}), the definition of the generalized relative spectral function in the form (\ref{rsa3}), and Stokes' theorem.

Actually we should prove only the relation (\ref{RSA}). For this purpose we need remember how formula (\ref{tildeta}) was obtained in the work \cite{ivvas} for the product structure case, when the metric tensor, $\mathcal{A}^-$, and $\mathcal{A}$ do not depend on $s$ near $\Sigma$.  The way is this:
\begin{enumerate}
	\item cut the manifold $\MM$ along the surface $\Sigma$ and paste the cylinder $\mathcal{C}_{\epsilon}=\Sigma\times[0,\epsilon]$, $\epsilon>0$;
	\item define on $\mathcal{C}_{\epsilon}$ special types of connection 1-forms (Riemannian and Yang--Mills) in a smooth way
	\begin{equation*}
	\widetilde{\Gamma}=\Gamma_{\mu}\big|_{s=0}dx^{\mu},\,\,\,\,
	\widetilde{\mathcal{A}}_{\epsilon}=\mathcal{A}^-_{\mu}\big|_{s=0}dx^{\mu}+f(s/\epsilon)
	\mathcal{A}_{\mu}\big|_{s=0}dx^{\mu},
	\end{equation*}
	where
	\begin{equation*}
	f\in C^{\infty}([0,1]):\,f(0)=0,\,f(1)=1,\,
	f^{(m)}(0)=f^{(m)}(1)=0\,\,\,\mbox{for all}\,\,\,m>0;
	\end{equation*}
	\item calculate a Pontryagin $n$-form $P_{\epsilon}$ in the $\mathcal{C}_{\epsilon}$ by using Seeley--DeWitt coefficients;
	\item calculate the limit $\tilde\eta(\DD^+,\DD^-)=-2\lim_{\epsilon\to+0}\int_{\mathcal{C}_{\epsilon}}P_{\epsilon}$ by using the fact that the last integral does not depend on the parameter $\epsilon$.
\end{enumerate}

Hence we can repeat points 3 and 4, using an equivalent formula for the Pontryagin form $P_{\epsilon}$ through the $\hat{A}$-genus and Chern polynomial $\mathrm{ch}$ (see \cite{Nakahara:2003nw})
\begin{equation}
P_{\epsilon}=\hat{A}(\widetilde{\mathcal{R}})\wedge\mathrm{ch}(\widetilde{\mathcal{F}}_{\epsilon}),
\end{equation}
where $\widetilde{\mathcal{R}}$ and $\widetilde{\mathcal{F}}_{\epsilon}$ are the corresponding curvatures for $\widetilde{\Gamma}$ and $\widetilde{\mathcal{A}}_{\epsilon}$.

Then, applying an auxiliary Lemma \ref{lem1} after choosing the parameters in the form
\begin{equation*}
V=\mathrm{ch},\,\,\,\,
\omega_1=\hat{A}(\widetilde{\mathcal{R}}),\,\,\,\,
\omega_2=0,
\end{equation*}
\begin{equation*}
\mathcal{B}_1=\mathcal{A}^-_{\mu}\big|_{s=0}dx^{\mu},\,\,\,\,
\mathcal{B}_2=0,\,\,\,\,
\mathcal{B}_3=\mathcal{A}_{\mu}\big|_{s=0}dx^{\mu},
\end{equation*}
we get an equivalent representation for the RSA in the form (\ref{rsa3}), from which the last relation (\ref{RSA}) of the theorem follows.
\end{proof}

Now we are ready to proceed to the second part and formulate a more general case. Let the Riemannian connection has different values, $\Gamma^-$ and $\Gamma^+$, from both sides of $\Sigma$, while the metric tensor remains continuous, as it was noted earlier. Then the Atiyah--Patodi--Singer index theorem for domain walls has the following form

\begin{theorem}
\label{th2}
Under the conditions described above we have the equality
\begin{equation}
\label{APST2}
\mathrm{Index}(\slashed{D})=\int_{\MM\backslash\Sigma}P - \tfrac 12 \tilde\eta(\mathcal{A}^+,\Gamma^+,\mathcal{A}^-,\Gamma^-).
\end{equation}
\end{theorem}
\begin{proof}
It can be achieved by using two cylinders and steps from the previous theorem. Indeed, let us cut the manifold $\MM$ along the surface $\Sigma$ and then paste two cylinders, $\mathcal{C}_{\epsilon_1}$ and $\mathcal{C}_{\epsilon_2}$, in a row. Then we define connection 1-forms on the cylinders in the following way:
\begin{enumerate}
	\item on $\mathcal{C}_{\epsilon_1}$ a Riemannian 1-form is defined in a smooth way between $\Gamma^-|_{\Sigma}$ and $\Gamma^+|_{\Sigma}$, while a Yang--Mills 1-form between $\mathcal{A}^-|_{\Sigma}$ and $\mathcal{A}^-|_{\Sigma}$;
	\item on $\mathcal{C}_{\epsilon_2}$ a Yang--Mills 1-form is defined in a smooth way between $\mathcal{A}^-|_{\Sigma}$ and $\mathcal{A}^+|_{\Sigma}$, while a Riemannian 1-form between $\Gamma^+|_{\Sigma}$ and $\Gamma^+|_{\Sigma}$.
\end{enumerate}

Then, using the considerations from the proof of the previous theorem, we need to find the limit
\begin{equation*}
-2\lim_{\epsilon_1\to+0}\int_{\mathcal{C}_{\epsilon_1}}P_{\epsilon_1}^1
-2\lim_{\epsilon_2\to+0}\int_{\mathcal{C}_{\epsilon_2}}P_{\epsilon_2}^2,
\end{equation*}
where $P_{\epsilon_1}^1$ and $P_{\epsilon_2}^2$ are the corresponding Pontryagin $n$-forms on $\mathcal{C}_{\epsilon_1}$ and $\mathcal{C}_{\epsilon_2}$. Let us also note that it is possible to construct the homotopy, under which all curvatures and their derivatives in the bulk, as well as the matching conditions on $\Sigma$ together with relevant geometric invariants on $\Sigma$ are smooth functions of homotopy parameter. Thus, the last integrals are also a smooth function of the parameter. Moreover they are constants, because the index is an integer.

Applying Lemma \ref{lem1} twice we obtain the function $\tilde\eta(\mathcal{A}^+,\Gamma^+,\mathcal{A}^-,\Gamma^-)$ in the form (\ref{rsa1}).
\end{proof}

\section{Auxiliary lemma}
\begin{lem}
\label{lem1}
Let $\epsilon_0>\epsilon>0$, and $\mathcal{A}_{\epsilon}=\mathcal{B}_1+s\mathcal{B}_2+f(s/\epsilon)\mathcal{B}_3$, where $f\in C^{\infty}([0,1])$, such that $f(0)=0$ and $f(1)=1$, is a family of connection 1-forms on a cylinder $\mathcal{C}_{\epsilon}=\Sigma\times[0,\epsilon]$. $\mathcal{F}_{\epsilon}$ is the corresponding field strength. Then let $V$ be an invariant polynomial, and $\omega=\omega_1+s\omega_2$ is a form of even degree, the density of which is smooth. Moreover, we require the densities of the forms $\mathcal{B}_1$, $\mathcal{B}_3$, and $\omega_1$ are smooth and do not depend on $s$ and $\epsilon$, while ones of the $\mathcal{B}_2$, $\dot{\mathcal{B}}_2$ and $\omega_2$ are smooth and bounded functions on $\mathcal{C}_{\epsilon}$ for all $\epsilon\in[0,\epsilon_0]$, where the dot denotes the derivative by $s$. If the following integral 
\begin{equation*}
\int_{\mathcal{C}_{\epsilon}}\omega\wedge V(\mathcal{F}_{\epsilon})
\end{equation*}
does not depend on $\epsilon$, then it is equal to
\begin{equation*}
\int_{\Sigma}\omega_1\wedge TV(\mathcal{B}_1+\mathcal{B}_3,\mathcal{B}_1),
\end{equation*}
where we mean the restriction of the forms on $\Sigma$, and $TV$ is the transgression of the polynomial $V$ in the formulation from \cite{Nakahara:2003nw} (see formula (11.21)).
\end{lem}
\begin{proof}
The main idea is to use the asymptotic expansion, when $\epsilon\to+0$. For this purpose we represent the field strength in the form of two terms
\begin{equation*}
\mathcal{F}_{\epsilon}=\epsilon^{-1}f^{\prime}(s/\epsilon)\,ds\wedge\mathcal{B}_3+\theta,
\end{equation*}
where the second form has a smooth density which is bounded for all $\epsilon\in[0,\epsilon_0]$. 

It is known that the decomposition $V=\sum_{k\geqslant0}V_{k}$, where $V_{k}$ is the polynomial of degree $k$, is possible. By the symbol $\widetilde{V}_{k}$ we denote the polarization of $V_{k}$. Using the fact that the form $ds$ can only appear once, we have the following expansion
\begin{equation*}
V(\mathcal{F}_{\epsilon})=
\epsilon^{-1}f^{\prime}(s/\epsilon)
\sum_{k\geqslant1}k\widetilde{V}_{k}
(ds\wedge\mathcal{B}_3,\mathcal{F}_{\epsilon},\ldots,\mathcal{F}_{\epsilon})+O(1).
\end{equation*}

After that, using an auxiliary field strength $\widehat{\mathcal{F}}_{s}$ for the 1-form 
$\mathcal{B}_1+s\mathcal{B}_3$
and doing the change of variables $s\to f(s/\epsilon)$ in the integral, we obtain
\begin{equation*}
\int_{\mathcal{C}_{1}}ds\wedge\omega_1\wedge\bigg[
\sum_{k\geqslant1}k\widetilde{V}_{k}
(\mathcal{B}_3,\widehat{\mathcal{F}}_{s},\ldots,\widehat{\mathcal{F}}_{s})\bigg]
+O(\epsilon).
\end{equation*}

The factor $ds$ is separated, so we can go to the restrictions of the forms
$\omega_1$, $\mathcal{B}_1$, and $\mathcal{B}_3$ on $\Sigma$. Furthermore, the first term does not depend on the parameter $\epsilon$, while the second one can be arbitrarily small and can be omitted. 

Also we know that $\omega_1$ does not depend on $s$. So, using the definition of the transgression, the statement of the lemma follows.
\end{proof}

\section{Discussion}
Despite the fact that formula (\ref{APST}) for the index of the Dirac operator $\slashed{D}$ has acquired a new concise form (\ref{APST1}), one of the issues discussed in the paper \cite{ivvas} remains. Namely, under what conditions do we have equality 
$\tilde\eta(\mathcal{A}^+,\Gamma,\mathcal{A}^-,\Gamma)=\tilde\eta(\DD^+,\DD^-)$? Several possible cases were listed in \cite{ivvas}, but they are not criteria. This work does not contain such kind of study.
In other words, we want to know a spectral function of which operator binds to $\tilde{\eta}(\mathcal{A}^+,\Gamma^+,\mathcal{A}^-,\Gamma^-)$ in the general case.

It is also worth noting that the Lemma \ref{lem1} is essentially an independent derivation of formula (\ref{APST1}) and provides its more elegant and concise proof. Moreover, we generalized the index theorem (see Theorem \ref{th2}) to the case when the Riemannian connection density has a jump. However, this latter formula is a little bit cumbersome.

\paragraph{Acknowledgements.}
The author thanks D.V. Vassilevich for discussion of the paper structures and application of the results. Also we express gratitude deeply indebted M.A. Semenov-Tian-Shansky for a careful reading of the manuscript and suggestion of amendments.
This research is fully supported by the grant in the form of subsidies from the Federal budget for state support of creation and development world-class research centers, including international mathematical centers and world-class research centers that perform research and development on the priorities of scientific and technological development. The agreement is between MES and PDMI RAS from \textquotedblleft8\textquotedblright\,November 2019 \textnumero\ 075-15-2019-1620. Also, A.V. Ivanov is a winner of the Young Russian Mathematician award and would like to thank its sponsors and jury.

\end{document}